\documentclass[11pt]{article}
\usepackage{amsmath,amssymb,amsthm,amsfonts,latexsym,bbm,xspace,graphicx,float,mathtools,mathdots,cancel}
\usepackage{braket,caption,subcaption,ellipsis,textcomp,MnSymbol}
\usepackage[colorlinks,citecolor=blue,bookmarks=true]{hyperref}
\usepackage[nameinlink]{cleveref}
\crefname{ineq}{inequality}{inequalities}
\usepackage{url}
\creflabelformat{ineq}{#2{\upshape(#1)}#3}
\usepackage[letterpaper,margin=1in]{geometry}
\usepackage{enumitem}

\usepackage{atbegshi}
\usepackage[explicit]{titlesec}
\usepackage[ruled,vlined,linesnumbered]{algorithm2e}
\newtheorem{theorem}{Theorem}[section]
\newtheorem{lemma}[theorem]{Lemma}

\newtheorem{conjecture}[theorem]{Conjecture}
\newtheorem{question}[theorem]{Question}
\newtheorem{problem}[theorem]{Problem}
\AddToHook{env/lemma/begin}
  {\crefalias{theorem}{lemma}}
\AddToHook{env/corollary/begin}
  {\crefalias{theorem}{corollary}}
\AddToHook{env/definition/begin}
  {\crefalias{theorem}{definition}}
\AddToHook{env/remark/begin}
  {\crefalias{theorem}{remark}}
\AddToHook{env/proposition/begin}
  {\crefalias{theorem}{proposition}}
\AddToHook{env/observation/begin}
  {\crefalias{theorem}{observation}}
\AddToHook{env/conjecture/begin}
  {\crefalias{theorem}{conjecture}}
\AddToHook{env/question/begin}
  {\crefalias{theorem}{question}}
\AddToHook{env/problem/begin}
  {\crefalias{theorem}{problem}}

\newcommand\sF{{\mathcal{F}}}
\newcommand\sfY{{\mathsf{Y}}}
\newcommand\sfZ{{\mathsf{Z}}}

\newcommand\E{{\mathbb{E}}}
\newcommand\N{{\mathbb{N}}}

\newcommand\remove[1]{{}}

\crefname{question}{question}{questions}
\crefname{problem}{problem}{problems}

\newcommand{\email}[1]{\href{mailto:#1}{\texttt{#1}}}

\title{An algorithm for $k$-set cover}
\author{Josh Alman\thanks{Columbia University. \email{josh@cs.columbia.edu}. Supported in part by NSF Grant CCF-2238221 and a Packard Foundation Fellowship.} \and
	Baitian Li\thanks{Columbia University. \email{bl3052@columbia.edu}. Supported in part by NSF Grant CCF-2238221, a Packard Foundation Fellowship, and a Columbia SEAS Presidential Fellowship.} \and
	Kevin Pratt\thanks{Columbia University. \email{ktp2116@columbia.edu}. Supported in part by NSF Grant CCF-2238221 and a Packard Foundation Fellowship.}}
\date{}

\begin{document}

\maketitle
\begin{abstract}
We show that set cover on a universe of size $n$ and with sets of size at most $k$ can be solved in time $2^{(1-1/k+O(1/k^{3/2}))n}$. This improves on a $2^{(1-0.929/k)n}$-time algorithm of Bj\"orklund (STACS 2010) for all sufficiently large $k$.
\end{abstract}

\section{Introduction}
In this paper we study algorithms for the following problem:
\begin{problem}[\textsc{$k$-set cover}]\label{prob:sc}
	Given as input an integer $t$ and a set family $\sF \subseteq 2^{[n]}$ consisting of sets of size at most $k$, decide if there exist at most $t$ sets in $\sF$ that cover $[n]$.
\end{problem}

Current algorithms for \Cref{prob:sc} run in time $O^*(2^{(1-f(k))n})$\footnote{We use $O^*$ to suppress polynomial factors in $n$.}, where $f(k) \to 0$ as $k \to \infty$ \cite{Koivisto09,bjorklund2010exact}. In \cite{SetCoverConj}, it was shown that if such a running time is inherent, meaning \textsc{$k$-set cover} cannot be solved in time $(2-\varepsilon)^n$ for some $\varepsilon$ independent of $k$, then our current fastest algorithms for a host of other problems, including subset sum, Steiner tree, and connected vertex cover, are essentially optimal. Notably, there is no known hardness of these problems under other popular conjectures in fine-grained complexity such as the Strong Exponential Time Hypothesis. Thus our interest is in the following conjecture:

\begin{conjecture}[Set cover conjecture]\label{conj:scc}
For all $\varepsilon > 0$, there exists $k \in \N$ such that $k$-set cover cannot be solved in time $O((2-\varepsilon)^n)$.
\end{conjecture}

One potential approach towards refuting \Cref{conj:scc} comes from a recent line of work \cite{bjorklund2024asymptotic, pratt2024stronger} which showed that Strassen's asymptotic rank conjecture, a generalization of the popular conjecture that $\omega = 2$, would imply that \textsc{$k$-set cover} can be solved in time $1.89^n$, therefore refuting \Cref{conj:scc}. This raises a basic question: what is the fastest \emph{unconditional} algorithm for \Cref{prob:sc}? That is, what is the asymptotically largest $f(k)$ such that we can solve $k$-set cover in time $2^{(1-f(k))n}$?

As discussed below, it turns out that current approaches yield $f(k) = C/k$ for some constant $C < 1$. Could we hope to improve this to $f(k) = (\log k) /k$, or even just $\omega(1)/k$? Less ambitiously, what about $100/k$? Even less ambitiously, what about $1/k$? We will show that at least this least-ambitious goal is (essentially) achievable.

\paragraph{Related work.}
The fastest algorithm explicitly stated for solving \Cref{prob:sc} is due to Koivisto \cite{Koivisto09}, and runs in deterministic time $O^*(2^{(1-1/(2k))n})$ and uses exponential space. This is the usual reference for the fastest algorithm for \textsc{$k$-set cover} (see e.g.~\cite{krauthgamer_et_al:LIPIcs.STACS.2019.45, nederlof2026invitation}). This algorithm can be understood as improving on a straightforward $2^n$-time dynamic program by reducing its state space.

In a later work, Bj\"orklund~\cite{bjorklund2010exact} designed a randomized algorithm for the related problem of \textsc{exact $k$-cover}, where all input sets have size exactly $k$ and we must find a partition of $[n]$. In fact, Bj\"orklund's algorithm may also be modified\footnote{\label{footnote}In short, one may replace $\sF$ with its downward closure to turn the cover into a partition, and find that the worst case in Bj\"orklund's algorithm has $k$-sized blocks, so this does not make the running time worse. The specific constant $0.92921\ldots$ comes from substituting the optimal $\tau_{12} = \sqrt{3}/2$, $\tau_2 = 1/2$ into \cite[Equation (4.11)]{bjorklund2010exact}. We do not elaborate further as we will achieve an improved running time here.} to solve the $k$-set cover problem in $O(2^{(1-0.9292/k)n})$ time and polynomial space. This algorithm works by reducing to the special case when the set family is $k$-partite, which is then solved in time $O^*(2^{(1-2/k)n})$ by an algebraic approach.

We also note that the complexity landscape of \textsc{$k$-set cover} resembles that of \textsc{$k$-SAT}, for which the fastest known algorithms run in time $2^{(1-(c-o_k(1))/k)n}$. The PPZ algorithm first achieved this form with $c=1$ \cite{PPZ99SAT}; Sch\"oning improved the constant to $c=\log_2 \mathrm{e}\approx 1.443$ \cite{Sch02SAT}; and the PPSZ algorithm further improved it to $c=\pi^2/6\approx 1.645$ \cite{PPSZ05SAT}. Obtaining an unbounded function of $k$ in place of $c$ remains an intriguing open problem, and would refute the ``Super-Strong Exponential Time Hypothesis'' \cite{VW19SSETH}.

\paragraph{Our results.}
We give a simple algorithm which improves on prior work for all large $k$:
\begin{theorem}\label{thm:main}
	There is a randomized algorithm solving \textsc{$k$-set cover} in time $O(2^{(1-\delta_k) n})$ and using exponential space, where $\delta_k = 1/k - O(1/k^{3/2})$.
\end{theorem}
Our algorithm essentially applies fast subset convolution \cite{bjorklund2007fourier} on certain random restrictions of the universe. A simple probabilistic argument will ensure that any solution is preserved by such a restriction with sufficiently high probability.

This leaves open the natural question:
\begin{question}
Can we solve \textsc{$k$-set cover} in time $2^{(1-1.00001/k)n}$?
\end{question}

\section{Preliminaries}

We use $H(p) = -p \log_2 p - (1-p) \log_2 (1-p)$ to denote the binary entropy function.

For a set family $\sF \subseteq 2^{[n]}$, we define its downward closure $\sF_\downarrow$ as the set family consisting of all subsets of sets in $\sF$:
\[\sF_\downarrow \coloneq \{X \subseteq [n] : X \subseteq Y \text{ for some } Y \in \sF\}.\]

Given two families $\sF_1, \sF_2 \subseteq 2^{[n]}$, let $\sF_1 \ostar \sF_2$ denote their \emph{subset convolution}:
\[\sF_1 \ostar \sF_2 \coloneq \{X \subseteq [n] : X = A \sqcup B \text{ for } A \in \sF_1, B \in \sF_2\}\]
We will use the following classical result of Bj\"orklund, Husfeldt, Kaski and Koivisto \cite{bjorklund2007fourier}; the subset convolution can be computed in near-optimal time.
\begin{theorem}[Fast subset convolution]\label{lem:fsc}
	Given as input $\sF_1, \sF_2 \subseteq 2^{[n]}$, we can compute $\sF_1 \ostar \sF_2$ in time $O^*(2^n)$.
\end{theorem}

\section{The algorithm}
\begin{proof}[Proof of \Cref{thm:main}]
Throughout the proof, $k$ is fixed and $n \to \infty$. We will only consider large-$k$ asymptotics at the very end after optimizing the running time. There is no harm in assuming that $t \le n$.

Fix $p \in (0,1)$ and let $q = 1-p$. We take $s_1 \coloneq q^kn$, $s_2 \coloneq s_3 \coloneq (1-q^k)n/2$, and $M = n \cdot \lceil 2^{-(1- q^k - pk)n/k + 2n^{2/3}+1}\rceil$.

\begin{algorithm}
\caption{An algorithm for $k$-set cover.}
\label{alg:k-set-cover}
Let $\sF_\downarrow$ be the downward closure of $\sF$

Sample $S \subseteq [n]$ by including each element independently with probability $p$

\If{$pn - n^{2/3} \le |S| \le pn + n^{2/3}$}
{
    Let $\mathcal{A}$ be a collection of $M$ uniformly random and independent subsets of $S$
    
	\ForEach{$t_1, t_2, t_3 \in \mathbb{Z}_{\geq 0}$ with $t_1 + t_2 + t_3 \leq t$}
	{
		Let $\sF_i \subseteq 2^{[n]}$ be the set family consisting of disjoint unions of $t_i$ sets in $\sF_\downarrow$, with total cardinality in the range $[s_i - 2n^{2/3}, s_i + 2n^{2/3}]$
	
		\ForEach{$A \in \mathcal{A}$}
		{
            $\sF_1' \gets \{U \setminus S: U \in \sF_1 , U \cap S = \emptyset \}$
            
            $\sF_2' \gets \{U \setminus S: U \in \sF_2, U \cap S = A\}$
            
            $\sF_3' \gets \{U \setminus S: U \in \sF_3, U \cap S = S \setminus A\}$
			
			\If{$[n] \setminus S \in \sF_1' \ostar \sF_2' \ostar \sF_3'$}
			{
				Accept.
			}
		}
	}
}
Reject.
\end{algorithm}

\paragraph{Correctness.}
We first claim that our algorithm has no false positives. Suppose it accepts; then there exist $P'_i \in \sF_i'$ for $i \in [3]$ with $P'_1 \sqcup P'_2 \sqcup P'_3 = [n] \setminus S$. By definition of $\sF_i'$, there then exist $P_i \in \sF_i$ with $P_1 \cap S = \emptyset$, $P_2 \cap S = A$, $P_3 \cap S = S \setminus A$. Therefore $P_1, P_2, P_3$ partition $[n]$, and by construction of $\sF_i$, this partition is a disjoint union of at most $t$ sets from the downward closure of $\sF$. This then lifts to a cover of $[n]$ with at most $t$ sets.

Conversely, suppose there exists a cover. Then there is a partition $[n] = \bigsqcup_{i=1}^\ell X_i$ into at most $t$ sets from the downward closure of $\sF$. Our algorithm will detect this partition whenever its parts admit a coarse partition, say $[\ell] = I_1 \sqcup I_2 \sqcup I_3$ with $P_i \coloneq \bigsqcup_{j \in I_i} X_j$, where
\[\bigl||P_i| - s_i\bigr| \le 2n^{2/3},\]
and where for some $A \in \mathcal{A}$, $P_1 \cap S = \emptyset, P_2 \cap S = A, P_3 \cap S = S \setminus A$. Indeed, the algorithm considers $(t_1,t_2,t_3)=(|I_1|,|I_2|,|I_3|)$, whose sum is $\ell\leq t$. We now show that with high probability, some such coarse partition will exist.

We consider two random variables, $\sfY$ and $\sfZ$. The variable $\sfY$ counts the total size of the parts of $[n]$ that are disjoint from $S$, whereas $\sfZ$ counts the number of parts that intersect $S$:
\[ \sfY = \sum_{i \in [\ell]} |X_i| \cdot \mathbf{1}[X_i \cap S = \emptyset], \quad \sfZ = \sum_{i \in [\ell]} \mathbf{1}[X_i \cap S \neq \emptyset]. \]

\begin{lemma}
	With probability at least $1 - \exp(-n^{\Omega(1)})$ over the random choice of $S$, all of the following hold:
	\begin{center}
		$\sfY \ge q^k n - n^{2/3}$, \quad $\sfZ \ge (1-q^k)(n/k) - n^{2/3}$ \quad and \quad $|S|\in [pn - n^{2/3}, pn + n^{2/3}]$.
	\end{center}
\end{lemma}
\begin{proof}
	By the union bound, it suffices to bound the failure probability of each event.
	For $\sfY$, since the event $X_i \cap S = \emptyset$ occurs with probability $q^{|X_i|}$,
	\[ \E[\sfY] = \sum_{i \in [\ell]} |X_i| \cdot q^{|X_i|} \geq \sum_{i \in [\ell]} |X_i| \cdot q^k = q^k n. \]
	Moreover, these events are independent, so Hoeffding's inequality gives

	\[ \Pr[\sfY < q^k n - n^{2/3}] \le \exp\left (\frac{-2n^{4/3}}{\sum_{i \in [\ell]} |X_i|^2}\right ) \le \exp\left (\frac{-2n^{4/3}}{nk}\right ) =  \exp\left(\frac{-2 n^{1/3}}{k}\right). \]

	For $\sfZ$, we have
	\[ \E[\sfZ] = \sum_{i\in [\ell]} (1-q^{|X_i|}) \ge \sum_{i\in [\ell]} \frac{|X_i|}{k}(1-q^k) = \frac{n}{k}(1-q^k) \]
	where we used the inequality $(1-q^j)/j \ge (1-q^k)/k$ for $j \le k$. Applying Hoeffding's inequality again gives
	\[ \Pr \left[\sfZ < (1-q^k)\frac{n}{k} - n^{2/3} \right] \le \exp \left( \frac{-2n^{4/3}}{\ell} \right) \le \exp(-2n^{1/3}). \]
	Finally, a standard Hoeffding bound shows that $\bigl||S| - pn\bigr| < n^{2/3}$ also holds with high probability.
\end{proof}

We now condition on these events holding.

Let $I \subseteq [\ell]$ be the set of indices of the parts that are disjoint from $S$.
The lower bound on $\sfY$ implies that there is a subcollection $J \subseteq I$ whose corresponding parts have total size in the range $[q^k n -n^{2/3}, q^k n -n^{2/3} + k]$. Let $P_1=\bigsqcup_{j\in J} X_j$ be their union. Then $P_1$ lies in the corresponding size range, so $P_1 \in \sF'_1$. We will distribute the remaining parts indexed by $I \setminus J$ between $P_2$ and $P_3$ later.

Every part indexed by $[\ell]\setminus I$ has a nonempty intersection with $S$. Let $R=\bigsqcup_{i\notin I} X_i$ be the union of these parts. Choose a random subcollection by including each such part independently with probability $1/2$, and let $T$ be its union. By Hoeffding's inequality,
\[\Pr[\bigl||T| - |R|/2\bigr| \ge n^{2/3}] \le 2 \exp(-2n^{1/3}/k).\]
This, combined with the fact that $\sfZ \ge (1-q^k)(n/k) -n^{2/3}$, shows that for sufficiently large $n$ there are at least $2^{(1-q^k)(n/k) -n^{2/3}-1}$ such sets $T$ with $\bigl||T| - |R|/2\bigr| < n^{2/3}$. Let $\mathcal{T}$ be the collection of all such $T$.

Because every part indexed by $[\ell]\setminus I$ intersects $S$, any two distinct $T,T' \in \mathcal{T}$ satisfy $T \cap S \neq T' \cap S$. It follows that the probability that a uniformly random $A \subseteq S$ satisfies $A = S \cap T$ for some $T \in \mathcal{T}$ is at least
\[\frac{2^{(1-q^k)(n/k) - n^{2/3}-1}}{2^{|S|}} \ge \frac{2^{(1-q^k)(n/k) - n^{2/3}-1}}{2^{pn + n^{2/3}}} = 2^{(1- q^k - pk)(n/k) - 2n^{2/3}-1}.\]
Since $\mathcal{A}$ contains $M = n \cdot \lceil 2^{-(1- q^k - pk)n/k + 2n^{2/3}+1}\rceil$ random sets, with probability at least $1-\mathrm{e}^{-n}$, some $A \in \mathcal{A}$ will satisfy $S \cap T = A$ for some $T \in \mathcal{T}$. For this $T$, since $\bigl| |T| - |R|/2 \bigr| < n^{2/3}$, we have
\[ \bigl| |T| - |R \setminus T| \bigr| < 2n^{2/3}. \]
We now show that the algorithm accepts for such an $A$. Starting with the two sides $T$ and $R \setminus T$, process the parts indexed by $I \setminus J$ one at a time, adding each part to the currently smaller side. If the current size difference is $d$, adding a part $X_i$ to the smaller side changes the difference to $\bigl|d-|X_i|\bigr| \leq \max\{d,k\}$. Thus, for sufficiently large $n$, the resulting unions $P_2$ and $P_3$ satisfy $\bigl||P_2|-|P_3|\bigr| < 2n^{2/3}$.

Together with $\bigl||P_1| - q^k n\bigr| \leq n^{2/3}$, this implies that $\bigl||P_2| - (1-q^k)(n/2)\bigr| \leq (3/2) n^{2/3}$; the same bound holds for $P_3$. Therefore this coarse partition satisfies the necessary conditions: $|P_i| \in [s_i \pm 2n^{2/3}]$, $P_1 \cap S = \emptyset$, $P_2 \cap S = A \in \mathcal{A}$, and $P_3 \cap S = S \setminus A$.

\paragraph{Running time.}
Assuming $q^k < 1/2$, the preprocessing step (line 6) takes time
\[ 2^{H((1-q^k)/2)n + o(n)} + 2^{H(q^k)n + o(n)}, \]
by computing for each $1 \le b \le t_i$ the subsets of $[n]$ of size at most $s_i + 2n^{2/3}$ that can be partitioned into $b$ sets from $\sF_\downarrow$ with a standard dynamic program. We then perform at most $|\mathcal{A}|\cdot (t+1)^3$ subset convolutions on the universe $[n] \setminus S$, each of which takes time $O^*(2^{n - |S|}) = 2^{n-pn + O(n^{2/3})}$. Note also that the construction of each $\sF_i'$ can be done in time $O^*(2^{n - |S|})$ by checking for each subset of $[n] \setminus S$ if its appropriate extension to $[n]$ belongs to $\sF_i$. Thus the final running time is
\begin{align*}
&2^{H((1-q^k)/2)n + o(n)} + 2^{H(q^k)n + o(n)}  + 2^{(-1+ q^k + pk)n/k + o(n)} \cdot 2^{n-pn+o(n)} \\
&= 2^{H((1-q^k)/2)n + o(n)} + 2^{H(q^k)n + o(n)} + 2^{(-1+ q^k + k)n/k + o(n)}.
\end{align*}

To optimize this, we choose $q = \left( \frac{2 \ln 2}{k}\right)^{1/(2k)}$. For $k \geq 6$, this choice satisfies $q \in (0,1)$ and $q^k=\sqrt{2\ln 2/k}<1/2$. The first term then has exponential rate
\[H((1-q^k)/2) < 1-q^{2k}/(2 \ln 2) = 1-1/k,\]
where we used the inequality $H((1-x)/2) < 1-x^2/(2 \ln 2)$.
The third term has rate
\[1-1/k + q^k/k= 1-1/k+(2 \ln 2)^{1/2}/k^{3/2}.\]
Moreover,
\[
H(q^k)=O(q^k\log(1/q^k))
=O((\log k)/\sqrt{k})=o_k(1),
\]
so the rate of the second term is negligible compared to the others for sufficiently large $k$. In conclusion, the third term dominates and any $\delta_k < 1/k - (2\ln 2)^{1/2} / k^{3/2}$ gives the desired time bound.
\end{proof}

\section*{AI Disclosure}

ChatGPT 5.6 was used to check for and correct typos in the writing, and to help verify our understanding of Bj\"orklund's algorithm described in footnote~\ref{footnote} above.

\bibliographystyle{alphaurl}

\bibliography{main.bib}

\end{document}